\documentclass[conference]{IEEEtran}
\IEEEoverridecommandlockouts
\usepackage{amsmath,amssymb,amsfonts}
\usepackage{algorithmic}
\usepackage{graphicx}
\usepackage{textcomp}
\usepackage{array}
\usepackage{diagbox}
\usepackage{bm}
\usepackage{stfloats}
\usepackage{xcolor}
\usepackage{romannum}
\usepackage{comment}
\usepackage{siunitx}
\definecolor{purple}{RGB}{128,0,128}

\newcommand{\new}{\textcolor{blue}}
\newcommand{\rev}{\textcolor{black}}

\newcommand{\chc}{\textcolor{black}}

\usepackage[style=ieee, doi=false, isbn=false, url=false]{biblatex}
\newtheorem{proposition}{{Proposition}}

\newenvironment{proof}{{\noindent\it Proof:}}{\hfill $\square$\par}

\begin{document}

\title{Channel State Information-Free Location-Privacy Enhancement: Delay-Angle Information Spoofing\\
} 
\vspace{-5pt}
\author{\IEEEauthorblockN{Jianxiu Li and
Urbashi Mitra}
{Department of Electrical and Computer Engineering, University of Southern California, CA, USA}\\
E-mail: \{jianxiul, ubli\}@usc.edu\thanks{\chc{This work has been funded by one or more of the following: the USC + Amazon Center on Secure and Trusted Machine Learning, DOE DE-SC0021417, Swedish Research Council 2018-04359, NSF CCF-2008927, NSF RINGS-2148313, NSF CCF-2200221, NSF CIF-2311653, ARO W911NF1910269, ONR 503400-78050, and ONR N00014-15-1-2550.}}\vspace{-5pt}}
\maketitle

\begin{abstract}
In this paper, a delay-angle information spoofing (DAIS) strategy is proposed for location-privacy enhancement. By shifting the location-relevant delays and angles without the aid of channel state information (CSI) at the transmitter, the eavesdropper is obfuscated by a physical location that is distinct from the true one. A precoder is designed to preserve location-privacy while the legitimate localizer can remove the obfuscation with the securely \textit{shared information}. Then, a lower bound on the localization error is derived via the analysis of the \textit{geometric mismatch} caused by DAIS, validating the enhanced
location-privacy. The statistical hardness for the estimation of the shared information is also investigated to assess the robustness to the potential leakage of the designed precoder structure. Numerical comparisons show that the proposed DAIS scheme results in more than $15$ dB performance degradation for the illegitimate localizer at high signal-to-noise ratios, which is comparable to a recently proposed CSI-free location-privacy enhancement strategy and is less sensitive to the precoder structure leakage than the prior approach. 
\end{abstract}


\vspace{-5pt}
\section{Introduction} \vspace{-3pt}
Thanks to the large bandwidth and limited multipath, millimeter wave (mmWave) signals \cite{Rappaport} have been widely employed to infer the location of user equipment (UE) in a multi-antenna system. \rev{If} the delay and angle information \rev{can be} precisely estimated from these wireless signals, centimeter-level localization accuracy is achievable with a single authorized device (AD)  \cite{Shahmansoori,Zhou,FascistaMISOML,Li}. However, how to preserve location-privacy is not considered in these designs \cite{Shahmansoori,Zhou,FascistaMISOML,Li} whose main focus is localization accuracy\rev{. D}ue to the nature of the propagation of electromagnetic waves, once these wireless signals are eavesdropped, the location of the UE is sensitive to \rev{exposure} to unauthorized devices (UD) that can \rev{potentially} further infer more private information, \textit{e.g.,} personal preference \cite{Ayyalasomayajula,li2023fpi}.

To limit the privacy leakage at the physical layer, the statistics of the channel, or the actual channel, has been leveraged for security designs \cite{Goel,Tomasin2,dacosta2023securecomm,Goztepesurveyprivacy,Checa, Ayyalasomayajula,li2023fpi,Ardagnaobfuscation}. Specific to protecting the location-relevant information from being easily snooped from the wireless signals, the prior strategies \cite{Goel,Tomasin2,Checa,Ayyalasomayajula} either inject artificial noise that is in the \textit{null space} of \rev{the} legitimate channel to decrease the received signal-to-noise ratio (SNR) for the UDs \cite{Goel,Tomasin2} or hide key delay and angle information via transmit beamforming to obfuscate the UDs \cite{Checa,Ayyalasomayajula}. However, all these designs in \cite{Goel,Tomasin2,Checa,Ayyalasomayajula} rely on accurate channel state information (CSI)\rev{;} acquiring such knowledge increases the overhead for resource-limited UEs.

{To enhance location-privacy without CSI, fake paths are designed in \cite{li2023fpi} and virtually injected to the channel via a precoding design as a form of jamming 
\cite{Li2,Li}. By virtue of the fake path injection (FPI) strategy \cite{li2023fpi}, accurately estimating the location-relevant parameters becomes statistically harder for an eavesdropper when the injected paths are highly correlated with the true paths. However, such a jamming design does not directly hide the location information itself, so location snooping is still possible at high SNRs especially when the bandwidth and the number of antennas are quite large. If the structure of the {precoder} is unfortunately leaked to the UD, the associated precoding matrix can be inferred with enough measurements, undermining} the efficacy of \cite{li2023fpi}.  Herein, motivated by the obfuscation technique in \cite{Ardagnaobfuscation}, we propose a delay-angle information spoofing (DAIS) strategy to enhance location-privacy without CSI. {DAIS \textit{virtually} moves the UE to an {\bf incorrect} location  which can be far from the true one. In contrast, the prior work \cite{li2023fpi} does not introduce \textit{geometric mismatch} though a challenging estimation framework is created. {As a result of the differences, in the current work, we develop a \textit{mismatched} Cram\'{e}r-Rao bound (MCRB) \cite{Fortunatimismatchsurvey} for the estimation error, versus a true Cram\'{e}r-Rao bound (CRB) in \cite{li2023fpi}}. Our theoretical analysis shows the amount of obfuscation possible via DAIS. This new strategy is also more robust to the leakage of the precoder structure.} 
The main contributions of this paper are: \vspace{-0.5pt}
\begin{enumerate}
\item A general framework is introduced to preserve location-privacy with DAIS, where all the location-relevant delays and angles are shifted without CSI such that the UD {is} obfuscated.
\item To spoof the UD with the shifted delays and angles, a {new} CSI-free precoding strategy is proposed, {distinct from \cite{li2023fpi}}; a design for the information  secretly shared with the AD is also provided to ensure performance.
\item {A MCRB is derived for DAIS, with a closed-form \chc{expression} for the pseudo-true (incorrect) locations, theoretically validating the efficacy of \chc{the} proposed scheme.} 
\item The impact of leaking the structure of the designed precoder to the UD is studied, which shows the robustness of the proposed scheme.
\item Numerical comparisons with the localization accuracy of the AD are provided, showing that more than $15$ dB performance degradation is achievable at high SNRs for the UD, due to the proposed DAIS method.
\vspace{-0.5pt}
\end{enumerate} 

We use the following notation. Scalars are denoted by lower-case letters $x$ and column vectors by bold letters $\bm{x}$. The $i$-th element of $\bm{x}$ is denoted by $\bm{x}[i]$. Matrices are denoted by bold capital letters $\bm X$ and $\boldsymbol{X}[i, j]$ is the ($i$, $j$)-th element of $\bm{X}$. The operators $|x|$, $\|\bm  x\|_{2}$, $\mathfrak{R}\{x\}$, $\mathfrak{I}\{x\}$, $\lfloor{x}\rfloor$, and $\operatorname{diag}(\mathcal{A})$ stand for the magnitude of $x$, the $\ell_2$ norm of $\bm x$, the real part of $x$, the imaginary part of $x$, the largest integer that is less than $x$, and a diagonal matrix whose diagonal elements are given by $\mathcal{A}$, respectively. $(x)_{(t_1,t_2]}$ with $t_1<t_2$ is defined as $(x)_{(t_1,t_2]}\triangleq x-\left\lfloor\frac{x-t_1}{t_2-t_1}\right\rfloor(t_2-t_1)$
and $\mathbb{E}\{\cdot\}$ \rev{denotes} the expectation of a random variable. The operators $\operatorname{Tr}(\cdot)$, $(\cdot)^{-1}$, $(\cdot)^\mathrm{T}$, and $(\cdot)^{\mathrm{H}}$ are defined as \rev{the trace, the inverse, the transpose, and the conjugate transpose of a vector or matrix, respectively}.

\vspace{-5pt}
\section{System Model}\label{sec:systemmodel}\vspace{-3pt}
\begin{figure}[t]
\centering
\includegraphics[scale=0.455]{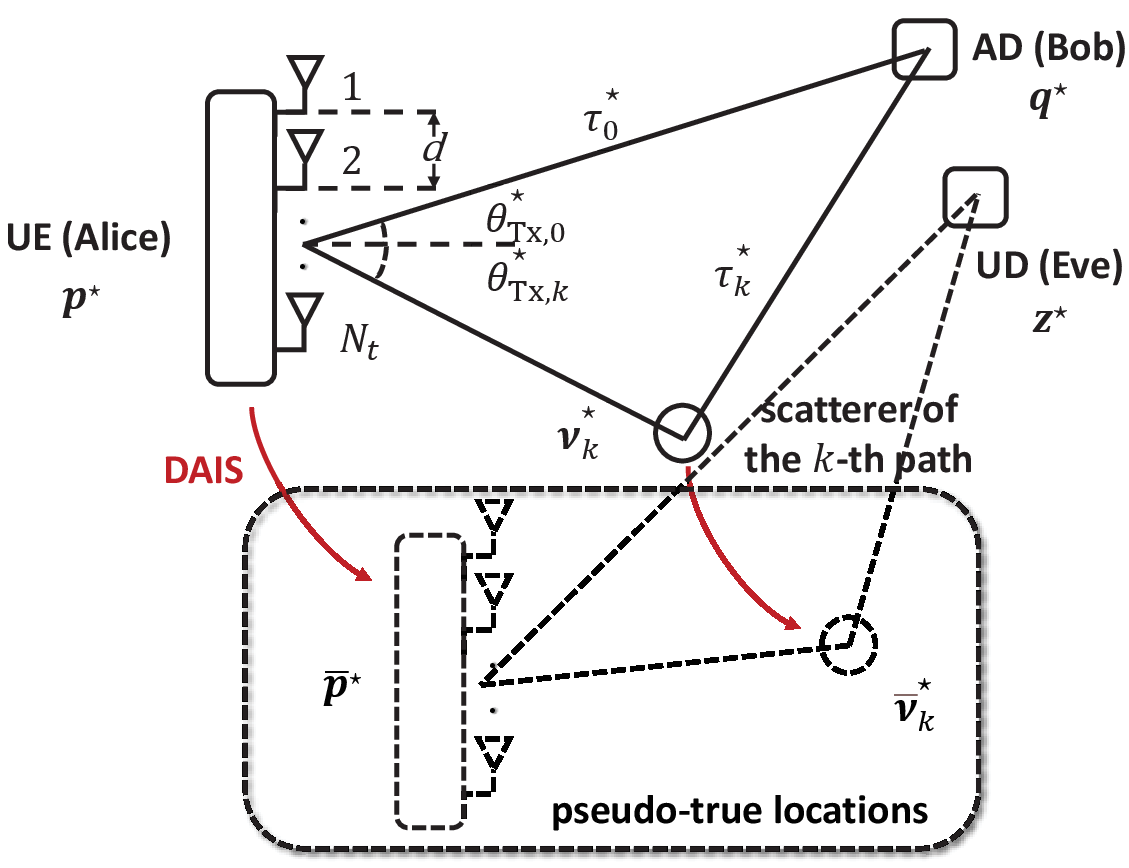}\vspace{-10pt}
\caption{\chc{System model.}}\vspace{-15pt}
\label{fig:sys}
\end{figure}
As shown in Figure \ref{fig:sys}, we consider a system model similar to \cite{li2023fpi}, where an \rev{AD} (Bob), at a location $\bm q^{\star}=[q^{\star}_x, q^{\star}_y]^{\mathrm{T}} \in \mathbb{R}^{2\times 1}$, serves a \rev{UE} (Alice) at an unknown position $\bm p^{\star} = [p^{\star}_x, p^{\star}_y]^{\mathrm{T}} \in \mathbb{R}^{2\times 1}$. To provide the location-based services, after Alice transmits pilot signals through a public channel, Bob can infer Alice's location \rev{from} the received signal. Unfortunately, there is an \rev{UD} (Eve), at a position $\bm z^{\star}=[z^{\star}_x, z^{\star}_y]^{\mathrm{T}}\in\mathbb{R}^{2\times 1}$, who can eavesdrop on the public channel to \rev{also} estimate Alice's location. We assume both Bob and Eve know the pilot signals as well as their own locations so Eve’s malicious inference jeopardizes Alice's location-privacy if no location-privacy preservation mechanisms are adopted.

\chc{Herein}, mmWave multiple-input-single-output (MISO) orthogonal frequency-division multiplexing (OFDM) signaling is used for the transmissions. Accordingly, Alice has $N_t$ antennas while both Bob and Eve are equipped \rev{with} a single antenna. Denoting by $N$ and $G$ the number of sub-carriers and the number of the transmitted signals, respectively, we express the $g$-th symbol transmitted over the $n$-th sub-carrier as $x^{(g,n)}$ and the corresponding beamforming vector as $\bm f^{(g,n)}\in\mathbb{C}^{N_t\times1}$ Then, the pilot signal can be written as $\boldsymbol{s}^{(g,n)}\triangleq \boldsymbol{ f}^{(g,n)}x^{(g,n)}\in\mathbb{C}^{N_t\times 1}
$ and the received signal is given by \vspace{-5pt}
\begin{equation}
{y}^{(g,n)}=\boldsymbol{h}^{(n)} \boldsymbol{s}^{(g,n)}+{w}^{(g,n)},\label{eq:rsignal}\vspace{-5pt}
\end{equation}
for $n = 0, 1, \cdots, N-1$ and $g = 1, 2, \cdots, G$, where $\bm h^{(n)}\in \mathbb{C}^{1\times N_t}$ is the $n$-th sub-carrier public channel vector while ${w}^{(g,n)}\sim \mathcal{CN}({0},\sigma^2)$ represents independent, zero-mean, complex Gaussian noise with variance $\sigma^2$.

We assume that there exist $K$ non-line-of-sight (NLOS) paths in the channel, apart from an available line-of-sight (LOS) path. The $k$-th NLOS path is produced by a scatterer at an unknown position $\bm v^{\star}_k = [v^{\star}_{k,x},v^{\star}_{k,y}]^{\mathrm{T}}\in\mathbb{R}^{2\times 1}$, with $k = 1,2,\cdots,K$. Denote by $c$, $\varphi_c$, $B$, and $T_s\triangleq\frac{1}{B}$, the speed of light, carrier frequency, bandwidth, and sampling period, respectively. A narrowband channel is considered in this paper, \textit{i.e.}, $B\ll \varphi_c$, and the public channel vector $\bm h^{(n)}$ can be modeled as \cite{li2023fpi,FascistaMISOMLCRB,FascistaMISOML}\vspace{-5pt}
\begin{equation}
\boldsymbol{h}^{(n)}\triangleq\sum_{k=0}^{K}\gamma^{\star}_k e^{\frac{-j 2\pi n\tau^{\star}_k}{N T_{s}}}\boldsymbol{ \alpha}\left(\theta^{\star}_{\mathrm{Tx},k}\right)^{\mathrm{H}},\vspace{-5pt}
\label{eq:channel_subcarrier}
\end{equation}
where $k=0$ corresponds to the LOS path while $\gamma^{\star}_k$, $\tau^{\star}_k$, and $\theta^{\star}_{\mathrm{Tx},k}$ represent the complex channel coefficient, the time-of-arrival (TOA), and the angle-of-departure (AOD) of the $k$-th path, respectively. The steering vector $\boldsymbol{ \alpha}\left(\theta^{\star}_{\mathrm{Tx},k}\right)\in\mathbb{C}^{N_t\times1}$ is defined as $\bm\alpha(\theta^{\star}_{\mathrm{Tx},k})\triangleq\left[1, e^{-j\frac{2\pi d\sin(\theta^{\star}_{\mathrm{Tx},k})}{\lambda_c}}, \cdots, e^{-j\frac{2\pi(N_t-1)d\sin(\theta^{\star}_{\mathrm{Tx},k})}{\lambda_c}}\right]^{\mathrm{T}}$, for $k=0,1,\cdots,K$, where $\lambda_c\triangleq\frac{c}{\varphi_c}$ is the wavelength and $d$ is the distance between antennas, designed as $d=\frac{\lambda_c}{2}$. Define $\bm v^{\star}_0\triangleq \bm q^{\star}$ (or $\bm v^{\star}_0\triangleq \bm z^{\star}$) for Bob (or Eve). From the geometry, the TOA and AOD of the $k$-th path are given by\footnote{It is assumed that the orientation angle of the antenna array and the clock bias are known to both Bob and Eve as prior information; without loss of generality, these parameters are set to zero in this paper.}\vspace{-5pt}
\begin{equation}\label{eq:geometry}
\begin{aligned}
\tau^{\star}_{k} &=\frac{\left\|\boldsymbol{v}^{\star}_0-\boldsymbol{v}^{\star}_{k}\right\|_{2} +\left\|\boldsymbol{p}^{\star}-\boldsymbol{v}^{\star}_{k}\right\|_{2}} {c}\\
\theta^{\star}_{\mathrm{Tx}, k} &=\arctan \left(\frac{v^{\star}_{k,y}-p^{\star}_{y}} {v^{\star}_{k,x}-p^{\star}_{x}}\right)\rev{,}
\end{aligned}\vspace{-5pt}
\end{equation} 
where we assume ${\tau^{\star}_k} \in(0,NT_s]$ and $\theta^{\star}_{\mathrm{Tx},k}\in(-\frac{\pi}{2},\frac{\pi}{2}]$ \cite{Li}. 

{Given the noise level characterized by $\sigma^2$}, once the received signals are collected, Alice's location can be inferred with the pilot signals. To enhance location-privacy, we aim to increase Eve’s localization error\rev{. Alice will transmit} the  shared information over a secure channel to maintain Bob's localization accuracy. Note that, CSI is assumed to be unavailable to Alice. 

\vspace{-5pt}
\section{Delay-Angle Information Spoofing for Location-Privacy Enhancement}\label{sec:csifreedesign}\vspace{-3pt}
For the model-based localization designs, such as {\cite{Shahmansoori,Li,FascistaMISOMLCRB,FascistaMISOML}}, channel parameters are typically estimated in the first stage and then the location is inferred from the location-relevant channel parameters, \textit{i.e.,} TOAs and AODs, according to the geometry given in Equation \eqref{eq:geometry}. Hence, high localization accuracy not only relies on super-resolution channel estimation \cite{Li}\rev{,} but also requires the knowledge of \rev{the} geometric model. To degrade Eve's localization accuracy, we propose to obfuscate Eve with a mismatched geometric model achieved via a \rev{delay-angle information spoofing} strategy.

\vspace{-5pt}
\subsection{Delay-Angle Information Spoofing}\label{sec:dis} \vspace{-2.5pt}
Let $\Delta_\tau$ and $\Delta_\theta$ represent two constants used for the DAIS design. To prevent Alice's location from being accurately inferred by Eve from the estimate of the location-relevant channel parameters, the TOAs and AODs of the paths in Eve’s channel are shifted according to $\Delta_\tau$ and $\Delta_\theta$, respectively, as \vspace{-5pt} 
\begin{equation}\label{eq:mismatchgeometry}
\begin{aligned}
\bar{\tau}_{k}& = \left(\tau^{\star}_{k} + \Delta_\tau\right)_{(0,NT_s]}\\
\bar{\theta}_{\mathrm{Tx}, k}&= \arcsin\left({\left(\sin(\theta^{\star}_{\mathrm{Tx}, k})+\sin(\Delta_\theta)\right)_{\left(-1,1\right]}}\right).
\end{aligned}\vspace{-5pt}
\end{equation} 

Though Eve is assumed to know the clock bias and the orientation angle, by shifting the TOAs and AODs during the transmission of the pilot signals, the proposed DAIS scheme misleads Eve into treating another physical location as the true one if she exploits the geometric model in Equation \eqref{eq:geometry} for localization. The degraded localization accuracy will be analyzed in Section \ref{sec:degradationdais}. Note that, since such shifts do not reply on the channel parameters, CSI is not needed. 
\vspace{-5pt}
\subsection{Precoder Design}\vspace{-2.5pt}

To protect location-privacy with the DAIS in Section \ref{sec:dis}, the mmWave MISO OFDM signaling is still employed but we design {a precoding matrix $\bm\Phi^{(n)}\in\mathbb{C}^{N_t\times N_t}$ as\footnote{{Though this new precoder is similar to part of the precoding matrix in \cite{li2023fpi}, defined as ${\bm \Phi}_{\text{FPI}}^{(n)} \triangleq\bm I_{N_t} + e^{-\frac{2\pi n \bar\delta_\tau}{NT_s}} \operatorname{diag}\left(\bm \alpha\left(\bar\delta_{\theta_{\text{TX}}}\right)^{\mathrm{H}}\right)$, where $\bm I\in\mathbb{R}^{N_t\times N_t}$ is an identity matrix while $\bar\delta_\tau$ and $\bar\delta_{\theta_{\text{TX}}}$ are two design parameters for FPI, the design of the parameters, \textit{i.e.,} $\Delta_\tau$, $\Delta_{\theta}$ versus $\bar\delta_\tau$, $\bar\delta_{\theta_{\text{TX}}}$, and the associated analyses in these two works are quite different.}}\vspace{-5pt}
\begin{equation}\label{eq:daisbeamformer}
    {\bm \Phi}^{(n)} \triangleq  e^{-\frac{\chc{j}2\pi n \Delta_\tau}{NT_s}} \operatorname{diag}\left(\bm \alpha\left(\Delta_\theta\right)^{\mathrm{H}}\right),\vspace{-5pt}
\end{equation}
for $n=0,1,\cdots,N-1$}. Then, through the public channel, the received signal can be re-expressed as \vspace{-5pt}
\begin{equation}
\begin{aligned}
\bar{y}^{(g,n)}
&=\boldsymbol{h}^{(n)}{{\bm \Phi}^{(n)}}{\boldsymbol{ s}}^{(g,n)}+{w}^{(g,n)},\\
&=\boldsymbol{h}^{(n)}e^{-\frac{\chc{j}2\pi n \Delta_\tau}{NT_s}} \operatorname{diag}\left(\bm \alpha\left(\Delta_\theta\right)^{\mathrm{H}}\right){\bm s}^{(g,n)}+{w}^{(g,n)}\\
&=\sum_{k=0}^{K}\gamma^{\star}_k e^{\frac{-j 2\pi n\bar{\tau}_k}{N T_{s}}}\boldsymbol{ \alpha}\left(\bar{\theta}_{\mathrm{Tx},k}\right)^{\mathrm{H}}{\bm s}^{(g,n)}+{w}^{(g,n)}\\
&=\bar{\boldsymbol{h}}^{(n)}{\bm s}^{(g,n)}+{w}^{(g,n)},
\end{aligned}\label{daisrsignal} \vspace{-5pt}
\end{equation}
where $\bar{\boldsymbol{h}}^{(n)}\triangleq\sum_{k=0}^{K}\gamma^{\star}_k e^{\frac{-j 2\pi n\bar{\tau}_k}{N T_{s}}}\boldsymbol{ \alpha}\left(\bar{\theta}_{\mathrm{Tx},k}\right)^{\mathrm{H}}$ represents a \textit{virtual channel} for the $n$-th sub-carrier, constructed based on the original channel ${\boldsymbol{h}}^{(n)}$ with shifted TOAs and AODs.

It will be shown in Section \ref{sec:eveerror} that shifting the TOAs and AODs virtually moves Alice and the $k$-th scatterer to other positions $\bar{\bm p}^\star\triangleq[\bar{p}^\star_x,\bar{p}^\star_y]^\mathrm{T}\in\mathbb{R}^{2\times1}$ and $\bar{\bm v}^\star_k\triangleq[\bar{v}^\star_{k,x},\bar{v}^\star_{k,y}]^\mathrm{T}\in\mathbb{R}^{2\times1}$, respectively, with $k=1,2,\cdots,K$. Therefore, due to the proposed DAIS scheme, after the channel estimation with the knowledge of $\bar{y}^{(g,n)}$ and ${\bm s}^{(g,n)}$, Eve cannot accurately infer Alice's true location using the mismatched geometric model in Equation \eqref{eq:geometry}. In contrast, since we assume that Bob receive\rev{s} the shared information $\bm\Delta\triangleq[\Delta_\tau,\Delta_\theta]^\mathrm{T}\in\mathbb{R}^{2\times1}$ through a secure channel that is inaccessible by Eve\footnote{Error in the shared information can reduce Bob's estimation accuracy as well\rev{,} but the study of such error is beyond the scope of this paper.}, Bob can construct effective pilot signals $\bar{\bm s}^{(g,n)}\triangleq {\bm\Phi^{(n)}}{\bm s}^{(g,n)}\in\mathbb{C}^{N_t\times1}$. By leveraging the original signal model in Equation \eqref{eq:rsignal} with the knowledge of $\bar{\bm s}^{(g,n)}$ for localization, Bob is not obfuscated by the proposed DAIS scheme and he can maintain his localization accuracy.

\vspace{-7pt}
\section{Localization Accuracy with Delay-Angle Information Spoofing}\label{sec:analysis}\vspace{-3pt}
\subsection{Effective Fisher Information for Channel Estimation}\vspace{-2.5pt}
Define $\bar{\bm\xi}\triangleq\left[\bar{\bm\tau}^\mathrm{T},\bar{\bm\theta}_{\mathrm{Tx}}^{\mathrm{T}},\mathfrak{R}\{\bm\gamma^{\star}\},\mathfrak{I}\{\bm\gamma^{\star}\}\right]^{\mathrm{T}}\in\mathbb{R}^{4(K+1)\times 1}$ as a vector of the unknown channel parameters, where $\bar{\bm\tau}\triangleq\left[\bar{\tau}_0, \bar{\tau}_1\cdots, \bar{\tau}_K\right]^{\mathrm{T}}\in\mathbb{R}^{(K+1)\times 1}$, $\bar{\bm\theta}_{\mathrm{Tx}}\triangleq\left[\bar{\theta}_{\mathrm{Tx},0}, \bar{\theta}_{\mathrm{Tx},1}\cdots, \bar{\theta}_{\mathrm{Tx},K}\right]^{\mathrm{T}}\in\mathbb{R}^{(K+1)\times 1}$, and $\bm\gamma^{\star} \triangleq [\gamma^{\star}_0,\gamma^{\star}_1,\cdots,\gamma^{\star}_K]^{\mathrm{T}}\in\mathbb{R}^{(K+1)\times 1}$. The Fisher information matrix (FIM) for the estimation of $\bar{\bm \xi}$, denoted as $\bm J_{\bar{\bm \xi}}\in\mathbb{R}^{4(K+1)\times 4(K+1)}$, is given by \cite{Scharf} \vspace{-5pt}
\begin{equation}\label{eq:FIMce}
    \bm J_{\bar{\bm \xi}} = \frac{2}{\sigma^2}\sum_{n=0}^{N-1}\sum_{g=1}^{G}\mathfrak{R}\left\{\left(\frac{\partial\bar{ u}^{(g,n)}}{\partial \bar{\bm\xi}}\right)^{*}\frac{\partial \bar{ u}^{(g,n)}}{\partial \bar{\bm\xi}}\right\}, \vspace{-5pt}
\end{equation}
where $\bar{u}^{(g,n)} \triangleq \bar{\boldsymbol{h}}^{(n)}\bm s^{(g,n)}$.

Let $\bar{\bm\eta}\triangleq[\bar{\bm\tau}^\mathrm{T},\bar{\bm\theta}_{\mathrm{Tx}}^{\mathrm{T}}]^{\mathrm{T}}\in\mathbb{R}^{2(K+1)\times1}
$ represent the location-relevant channel parameters and we partition the FIM $\bm J_{\bar{\bm \xi}}$ as $\bm J_{\bar{\bm \xi}}= \begin{bmatrix}\bm J_{\bar{\bm \xi}}^{(1)} &\bm J_{\bar{\bm \xi}}^{(2)}\\\bm J_{\bar{\bm \xi}}^{(3)}&\bm J_{\bar{\bm \xi}}^{(4)}\end{bmatrix}$, with $\bm J_{\bar{\bm \xi}}^{(m)}\in\mathbb{R}^{2(K+1)\times2(K+1)}$, for $m=1,2,3,4$. To analyze the localization accuracy, the channel coefficients are considered as nuisance parameters and accordingly the effective FIM for the estimation of the location-relevant channel parameters $\bar{\bm\eta}$ can be derived as \cite{Tichavskyefim}\vspace{-5pt}
\begin{equation} \label{eq:efim}
    \bm J_{\bar{\bm \eta}} = \bm J_{\bar{\bm \xi}}^{(1)}-\bm J_{\bar{\bm \xi}}^{(2)}\left(\bm J_{\bar{\bm \xi}}^{(4)}\right)^{-1}\bm J_{\bar{\bm \xi}}^{(3)}\in\mathbb{R}^{2(K+1)\times 2(K+1)}. \vspace{-5pt}
\end{equation}
Using the proposed DAIS method for location-privacy enhancement, the localization accuracy for Bob and Eve will be studied in the following subsections, respectively.


\vspace{-6pt}
\subsection{Bob's Localization Error} \vspace{-2.5pt}
Since Bob has the access to the shared information $\bm\Delta$, similar to Equation \eqref{eq:efim}, the effective FIM for the estimation of $\bm\eta^\star\triangleq[(\bm\tau^\star)^{\mathrm{T}},(\bm\theta^\star_{\mathrm{Tx}})^{\mathrm{T}}]\in\mathbb{R}^{2(K+1)\times 1}$ with ${\bm\tau}^\star\triangleq\left[{\tau}^\star_0, {\tau}^\star_1\cdots, {\tau}^\star_K\right]^{\mathrm{T}}\in\mathbb{R}^{(K+1)\times 1}$ and ${\bm\theta}^\star_{\mathrm{Tx}}\triangleq\left[{\theta}^\star_{\mathrm{Tx},0}, {\theta}^\star_{\mathrm{Tx},1}\cdots, {\theta}^\star_{\mathrm{Tx},K}\right]^{\mathrm{T}}\in\mathbb{R}^{(K+1)\times 1}$ can be derived, which is denoted as $\bm J_{{\bm \eta}^\star}\in\mathbb{R}^{2(K+1)\times 2(K+1)}$. Then, the FIM for Bob's localization, denoted as $\bm J_{\bm \phi^{\star}}\in\mathbb{R}^{2(K+1)\times 2(K+1)}$, is given by \vspace{-5pt}
\begin{equation}\label{eq:FIMloc}
    \bm J_{\bm \phi^{\star}} = \bm\Pi_{\bm \phi^{\star}}^\mathrm{T}\bm J_{{\bm \eta}^\star}\bm\Pi_{\bm \phi^{\star}} , \vspace{-5pt}
\end{equation}
where $\bm\phi^{\star}\triangleq [(\boldsymbol{p}^{\star})^{\mathrm{T}},(\boldsymbol{v}^{\star}_1)^{\mathrm{T}},(\boldsymbol{v}^{\star}_2)^{\mathrm{T}},\cdots,(\boldsymbol{v}^{\star}_K)^{\mathrm{T}}]^{\mathrm{T}}\in\mathbb{R}^{2(K+1)\times1}$ is a vector of the true locations of Alice and scatterers while $\bm\Pi_{\bm \phi^{\star}}\triangleq\frac{\partial {\bm\eta}^\star}{\partial\bm\phi^{\star}}\in\mathbb{R}^{2(K+1)\times2(K+1)}$ can be derived according to the true geometric model in Equation \eqref{eq:mismatchgeometry}. Let $\hat{\bm\phi}_{\text{Bob}}$ be Bob's estimate of $\bm\phi^{\star}$ with an \textit{unbiased estimator}. The mean squared error (MSE) of such an estimator can be bounded as follows \cite{Scharf} \vspace{-3pt}
\begin{equation}\label{eq:BobCRB}
    \mathbb{E}\left\{\left(\hat{\bm\phi}_{\text{Bob}}-{\bm\phi}^{\star}\right)\left(\hat{\bm\phi}_{\text{Bob}}-{\bm\phi}^{\star}\right)^{\mathrm{T}}\right\}\succeq\bm\Xi_{\bm \phi^{\star}}\triangleq{\bm J_{\bm \phi^{\star}}^{-1}}, \vspace{-3pt}
\end{equation}
which is the well-known CRB.

\vspace{-6pt}
\subsection{Eve's Localization Error} \vspace{-2.5pt}\label{sec:eveerror}
To analyze Eve's localization accuracy, employing an \textit{efficient estimator} for channel estimation, we denote by $\hat{\bm\eta}_{\text{Eve}}$ Eve’s estimate of $\bar{\bm\eta}$, and assume the true distribution of $\hat{\bm\eta}_{\text{Eve}}$ as in \cite{ZhengmismatchRIS}, \textit{i.e.,} \vspace{-3.5pt}
\begin{equation}\label{eq:truemodel}
\hat{\bm\eta}_{\text{Eve}}=u(\bm\phi^\star)+\bm\epsilon,\vspace{-5pt}
\end{equation}
where $\bm\epsilon$ is a zero-mean Gaussian random vector with covariance matrix $\bm\Sigma_{\bar{\bm\eta}}\triangleq\bm J_{\bar{\bm\eta}}^{-1}$. Herein, we have $\bar{\bm\eta}=u(\bm\phi^\star)$ with $u(\cdot)$ being a function mapping the location \rev{information} $\bm\phi^\star$ to the location-relevant channel parameters $\bar{\bm\eta}$ according to the true geometric model defined in Equation \eqref{eq:mismatchgeometry}. However, since the shared information is unavailable to Eve, for a given vector of potential locations of Alice and scatterer $\bar{\bm\phi}\triangleq [(\bar{\boldsymbol{p}})^{\mathrm{T}},(\bar{\boldsymbol{v}}_1)^{\mathrm{T}},(\bar{\boldsymbol{v}}_2)^{\mathrm{T}},\cdots,(\bar{\boldsymbol{v}}_K)^{\mathrm{T}}]^{\mathrm{T}}\in\mathbb{R}^{2(K+1)\times1}$, \rev{Eve} will incorrectly believe that the estimates of channel parameters $\hat{\bm\eta}_{\text{Eve}}$ are modelled as \vspace{-5pt}
\begin{equation}\label{eq:mismatchmodel}
\hat{\bm\eta}_{\text{Eve}}=o(\bar{\bm\phi})+\bm\epsilon, \vspace{-5pt}
\end{equation}
where $o(\cdot)$ is function similar to $u(\cdot)$\rev{,} but is defined according to the mismatched geometric model assumed in Equation \eqref{eq:geometry}. The true and mismatched distributions of $\hat{\bm\eta}_{\text{Eve}}$ are denoted as $g_{\text{T}}(\hat{\bm\eta}_{\text{Eve}}|\bm\phi^\star)$ and $g_{\text{M}}(\hat{\bm\eta}_{\text{Eve}}|\bar{\bm\phi})$, respectively. Then we can find a vector of the pseudo-true locations of Alice and scatterers $\bar{\bm\phi}^\star\triangleq [(\bar{\boldsymbol{p}}^\star)^{\mathrm{T}},(\bar{\boldsymbol{v}}^\star_1)^{\mathrm{T}},(\bar{\boldsymbol{v}}^\star_2)^{\mathrm{T}},\cdots,(\bar{\boldsymbol{v}}^\star_K)^{\mathrm{T}}]^{\mathrm{T}}\in\mathbb{R}^{2(K+1)\times1}$, when Eve exploits Equation \eqref{eq:mismatchmodel} for localization though the true distribution of $\hat{\bm\eta}_{\text{Eve}}$ is defined according to Equation \eqref{eq:truemodel}, such that \cite{Fortunatimismatchsurvey} \vspace{-5pt}
\begin{equation}\label{eq:KLD}
    \bar{\bm\phi}^\star = \arg\min_{\bar{\bm\phi}} D\left(g_{\text{T}}(\hat{\bm\eta}_{\text{Eve}}|\bm\phi^\star)\|g_{\text{M}}(\hat{\bm\eta}_{\text{Eve}}|\bar{\bm\phi})\right),\vspace{-5pt}
\end{equation}
where $D(\cdot\|\cdot)$ represents the Kullback–Leibler (KL) divergence for two given distributions. The closed-form expression of $\bar{\bm\phi}^\star$ will be derived later in this subsection. 

Denote by $\hat{\bm \phi}_\text{Eve}$ a \textit{misspecified-unbiased} estimator designed according to the mismatched model in Equation \eqref{eq:mismatchmodel}. With the DAIS for location-privacy enhancement, there is a lower bound for the MSE of Eve’s localization based on the analysis of the MCRB \cite{Fortunatimismatchsurvey} \vspace{-5pt}
\begin{equation}\label{eq:EveLB}
\begin{aligned}
    &\mathbb{E}\left\{\left(\hat{\bm \phi}_\text{Eve}-{\bm \phi}^{\star}\right)\left(\hat{\bm \phi}_\text{Eve}-{\bm \phi}^{\star}\right)^{\mathrm{T}}\right\}\\
    &\quad\succeq \bm\Psi_{\bar{\bm \phi}^\star}\triangleq{\underbrace{\bm A_{\bar{\bm \phi}^\star}^{-1}\bm B_{\bar{\bm \phi}^\star}\bm A_{\bar{\bm \phi}^\star}^{-1}}_{\bm\Psi_{\bar{\bm \phi}^\star}^{(\romannum{1})}}+\underbrace{(\bar{\bm\phi}^\star-\bm\phi^{\star})(\bar{\bm\phi}^\star-\bm\phi^{\star})^\mathrm{T}}_{\bm\Psi_{\bar{\bm \phi}^\star}^{(\romannum{2})}}}, \vspace{-5pt}
\end{aligned}
\end{equation}
where $\bm A_{\bar{\bm \phi}^\star}\in\mathbb{R}^{2(K+1)\times2(K+1)}$ and $\bm B_{\bar{\bm \phi}^\star}\in\mathbb{R}^{2(K+1)\times2(K+1)}$ are two generalized FIMs, defined as \vspace{-5pt}
\begin{equation}\label{eq:A}
    \bm A_{\bar{\bm \phi}^\star}[r,l]\triangleq\mathbb{E}_{g_{\text{T}}(\hat{\bm\eta}_{\text{Eve}}|\bm\phi^\star)}\left\{\frac{\partial^2}{\partial\bar{\bm\phi}^\star[r]\partial\bar{\bm\phi^\star}[l]}\log g_{\text{M}}(\hat{\bm\eta}_{\text{Eve}}|\bar{\bm\phi}^\star)\right\},
\end{equation}
and
\begin{equation}\label{eq:B}
\begin{aligned}
&\bm B_{\bar{\bm \phi}^\star}[r,l]\\
&\triangleq\mathbb{E}_{g_{\text{T}}(\hat{\bm\eta}_{\text{Eve}}|\bm\phi^\star)}\left\{\frac{\partial\log g_{\text{M}}(\hat{\bm\eta}_{\text{Eve}}|\bar{\bm\phi}^\star) }{\partial\bar{\bm\phi}^\star[r]}\frac{\partial\log g_{\text{M}}(\hat{\bm\eta}_{\text{Eve}}|\bar{\bm\phi}^\star) }{\partial\bar{\bm\phi}^\star[l]}\right\},
\end{aligned}
\end{equation}
for $r,l = 1,2,\cdots,2(K+1)$. 

The next step is to derive the closed-form expression of $\bar{\bm\phi}^{\star}$. Since $\bar{\bm\phi}^{\star}$ is the parameter vector that minimizes \rev{the} KL divergence for two given Gaussian distributions $g_{\text{T}}(\hat{\bm\eta}_{\text{Eve}}|\bm\phi^\star)$ and $g_{\text{M}}(\hat{\bm\eta}_{\text{Eve}}|\bar{\bm\phi})$ according to Equation \eqref{eq:KLD}, if there exists a unique vector $\bar{\bm\phi}^\prime\in\mathbb{R}^{2(K+1)\times1}$ such that $o(\bar{\bm\phi}^\prime) = u(\bm\phi^\star)$, we have $\bar{\bm\phi}^\star=\bar{\bm\phi}^\prime$ according to the non-negative property of KL divergence\footnote{An alternative proof of this statement can be found in \cite[Equations (13) and (14)]{ZhengmismatchRIS}, in terms of the KL divergence for two Gaussian distributions.}. Then, \rev{our} goal amounts to deriving such a $\bar{\bm\phi}^\prime$. Since the path with the smallest TOA is typically treated as the LOS path for localization \cite{Ayyalasomayajula,FascistaMISOML}, considering the effect of the potential \textit{phase wrapping} caused by the proposed DAIS scheme, we discuss the following two cases.

\begin{enumerate}
    \item[C1)] $\bar{\tau}_0=\min\{\bar{\tau}_0, \bar{\tau}_1,\cdots \bar{\tau}_K\}$ holds. In \rev{the} case, the LOS path can be correctly distinguished if error in the estimated TOAs is small enough. By solving $o(\bar{\bm\phi}) = u(\bm\phi^{\star})$ for $\bar{\bm\phi}^{\star}$, we have \rev{the} unique solution as \vspace{-5pt}

        \begin{equation}\label{eq:ptruecase1}
            \begin{aligned}
                \bar{\boldsymbol{p}}^{\star} &= \boldsymbol{z} - c\chc{\bar\tau}_0[\cos(\chc{\bar\theta}_{\mathrm{Tx},0}),\sin(\chc{\bar\theta}_{\mathrm{Tx},0})]^{\mathrm{T}}\\
                \bar{v}^\star_{k,x} &=\frac{1}{2}\bar{b}^{\star}_{k}\cos(\chc{\bar{\theta}}_{\mathrm{Tx},k})+\bar{p}^{\star}_{x}\\
                \bar{v}^\star_{k,y} &= \frac{1}{2}\bar{b}^{\star}_{k}\sin(\chc{\bar{\theta}}_{\mathrm{Tx},k})+\bar{p}^{\star}_{y},
            \end{aligned}\vspace{-5pt}
        \end{equation}
        where \vspace{-3pt}
        \begin{equation}
           \begin{aligned}    
           &\bar{b}^{\star}_{k}\\
           &=\frac{\left(c\chc{\bar\tau}_k\right)^2-\left(z_x-\bar{p}^{\star}_{x}\right)^2-\left(z_y-\bar{p}^{\star}_{y}\right)^2}{c\chc{\bar\tau}_k-\left(z_x-\bar{p}^{\star}_{x}\right)\cos(\chc{\bar{\theta}}_{\mathrm{Tx},k})-\left(z_y-\bar{p}^{\star}_{y}\right)\sin(\chc{\bar{\theta}}_{\mathrm{Tx},k})}, 
           \end{aligned}\vspace{-2pt}
        \end{equation}
        with $k=1,2,\cdots,K$.
 
    \item[C2)] $\bar{\tau}_0\neq\min\{\bar{\tau}_0, \bar{\tau}_1,\cdots \bar{\tau}_K\}$ holds. In \rev{the} case, the LOS path cannot be correctly distinguished even with the true TOAs\footnote{From the aspect of obfuscating the LOS path, the proposed DAIS scheme can be considered as an extension of \cite{Ayyalasomayajula}; the design in \cite{Ayyalasomayajula} relies on CSI that is not needed in our scheme.}. Without loss of generality, we assume that $\bar{\tau}_1=\min\{\bar{\tau}_0, \bar{\tau}_1,\cdots \bar{\tau}_K\}$. By solving $o(\bar{\bm\phi}) = u(\bm\phi^{\star})$ for $\bar{\bm\phi}^{\star}$, the pseudo-true locations of Alice and the first scatterer are given by \vspace{-5pt}
        \begin{equation}\label{eq:ptruecase2}
            \begin{aligned}
                \bar{\boldsymbol{p}}^{\star} &= \boldsymbol{z} - c\chc{\bar\tau}_1[\cos(\chc{\bar\theta}_{\mathrm{Tx},1}),\sin(\chc{\bar\theta}_{\mathrm{Tx},1})]^{\mathrm{T}}\\
                \bar{v}^\star_{1,x} &=\frac{1}{2}\bar{b}^{\star}_{0}\cos(\chc{\bar{\theta}}_{\mathrm{Tx},0})+\bar{p}^{\star}_{x}\\
                \bar{v}^\star_{1,y} &= \frac{1}{2}\bar{b}^{\star}_{0}\sin(\chc{\bar{\theta}}_{\mathrm{Tx},0})+\bar{p}^{\star}_{y}. 
            \end{aligned} \vspace{-5pt}
        \end{equation}
    The pseudo-true locations of the other scatterers are the same \rev{as those} derived in Case C1. 
\end{enumerate}
{We note that the lower bound on Eve's estimation error in \cite[Equation (22)]{li2023fpi} is derived from the analysis of the CRB in the presence of the FPI, since Eve will believe that there are more paths in her channel given the \textit{geometric feasibility} of the injected fake paths \cite{li2023fpi}.  {In contrast, herein we introduce {geometric mismatch}, thus the results in \cite{li2023fpi} cannot characterize the performance degradation caused by DAIS as revealed in Equation \eqref{eq:EveLB}.}} 
\vspace{-5pt}
\subsection{Degraded Localization Accuracy}\label{sec:degradationdais}\vspace{-1.5pt}
By comparing the \rev{lower bounds} of Bob's and Eve's localization error, derived in Equations \eqref{eq:BobCRB} and \eqref{eq:EveLB}, respectively, we can show that the proposed DAIS strategy can \rev{effectively} decrease Eve's localization accuracy as follows.

\begin{proposition}\label{prop:degradation} Supposed that Equations \eqref{eq:truemodel} and \eqref{eq:mismatchmodel} characterize the true and mismatched distributions of the estimated channel parameters $\hat{\bm\eta}_{\text{Eve}}$, there exists a constant $\sigma_0$ such that $\operatorname{Tr}(\bm\Psi_{\bar{\bm \phi}^\star})\geq\operatorname{Tr}(\bm\Xi_{\bm \phi^{\star}})$  when $0\leq\sigma\leq\sigma_0$, \rev{where $\sigma$ represents the standard deviation of the Gaussian noise ${w}^{(g,n)}$ while $\bm\Xi_{\bm \phi^{\star}}$ and $\bm\Psi_{\bar{\bm \phi}^\star}$ are defined in Equations \eqref{eq:BobCRB} and $\eqref{eq:EveLB}$, respectively.}
\end{proposition}
\begin{proof}
    For the given true and mismatched distributions $g_{\text{T}}(\hat{\bm\eta}_{\text{Eve}}|\bm\phi^\star)$ and $g_{\text{M}}(\hat{\bm\eta}_{\text{Eve}}|\bar{\bm\phi})$ characterized by  Equations \eqref{eq:truemodel} and \eqref{eq:mismatchmodel}, it can be verified that $\bm\Psi_{\bar{\bm \phi}^\star}^{(\romannum{1})}$ and $\bm\Psi_{\bar{\bm \phi}^\star}^{(\romannum{2})}$ are positive semidefinite matrices 
    according to Equations \eqref{eq:EveLB}, \eqref{eq:A} and \eqref{eq:B}, while $\operatorname{Tr}\left(\bm\Psi^{(\romannum{2})}_{\bar{\bm \phi}^\star}\right)$ is not relevant {to} $\sigma$ \rev{from} Equations \eqref{eq:ptruecase1} and \eqref{eq:ptruecase2}. Hence, $\operatorname{Tr}\left(\bm\Psi^{(\romannum{1})}_{\bar{\bm \phi}^\star}\right)\geq0$ and $\operatorname{Tr}\left(\bm\Psi^{(\romannum{2})}_{\bar{\bm \phi}^\star}\right)\geq0$ hold and the goal amounts to proving that there is a constant $\sigma_0$ such that $\operatorname{Tr}\left(\bm\Psi^{(\romannum{2})}_{\bar{\bm \phi}^\star}\right)>\operatorname{Tr}(\bm\Xi_{\bm \phi^{\star}})$ for any $0\leq\sigma\leq\sigma_0$. Then, from Equations \eqref{eq:FIMce}, \eqref{eq:efim}, and \eqref{eq:FIMloc}, we have $\lim_{\sigma\downarrow0}\operatorname{Tr}(\bm\Xi_{\bm \phi^{\star}})=0$, yielding the desired statement.
\end{proof}

According to Equations \eqref{eq:BobCRB} and \eqref{eq:EveLB}, Proposition \ref{prop:degradation} shows that Eve cannot estimate Alice's location more accurately than Bob if the value of $\sigma$ is small enough. {Considering that $\lim_{\sigma\downarrow0}\operatorname{Tr}(\bm\Psi_{\bar{\bm \phi}^\star}^{(\romannum{1})})=0$ also holds}, the {geometric mismatch} introduced by the proposed DAIS scheme, {which corresponds to the quantity $\operatorname{Tr}(\bm\Psi_{\bar{\bm \phi}^\star}^{(\romannum{2})})$,  is dominant} in the degradation of Eve's localization accuracy at a high SNR. We will numerically show the impact
of the choices of the design parameter $\bm\Delta$ on Eve’s localization accuracy in Section \ref{sec:numericalresults}.
\vspace{-6pt}
\subsection{Precoder Structure Leakage}\label{sec:strucleakage}\vspace{-2.5pt}
If the {structure of the precoder} designed in Equation \eqref{eq:daisbeamformer} is leaked, Eve {will}  endeavor to estimate $\bm\Delta$ to avoid the geometric mismatch, yet \rev{this} is a statistically hard estimation problem if $\Delta$ is an unknown deterministic vector according to the following proposition. 

\begin{proposition}\label{prop:strucleakge} Assume that $\Delta$ is an unknown deterministic vector. Let ${\bm\chi}\triangleq\left[({\bm\tau}^\star)^\mathrm{T},({\bm\theta}_{\mathrm{Tx}}^\star)^{\mathrm{T}},\mathfrak{R}\{(\bm\gamma^{\star})^\mathrm{T}\},\mathfrak{I}\{(\bm\gamma^{\star})^\mathrm{T}\},\Delta^\mathrm{T}\right]^{\mathrm{T}}\in\mathbb{R}^{(4K+6)\times 1}$ and $\bm J_{{\bm \chi}}\in\mathbb{R}^{(4K+6)\times (4K+6)}$ be a vector of the unknown channel parameters and the associated FIM, respectively. $\bm J_{{\bm \chi}}$ is a singular matrix.
\end{proposition}
\begin{proof}
    Considering the knowledge of the {structure of the precoder}, we denote by ${u}^{(g,n)} \triangleq\boldsymbol{h}^{(n)}{{\bm \Phi}^{(n)}}{\boldsymbol{s}}^{(g,n)}$ the noise-free observation for the estimation of $\bm\chi$, with $g=1,2,\cdots,G$ and $n=0,1,\cdots,N-1$. It can be verified that $\frac{\partial  u^{(g,n)}}{\partial \Delta_\tau}= \sum_{k=0}^{K} \frac{\partial u^{(g,n)}}{\partial {\tau}^\star_k}$, and $\frac{\partial  u^{(g,n)}}{\partial \Delta_{\theta}}= \sum_{k=0}^{K} \frac{\partial u^{(g,n)}}{\partial {\theta}^\star_{\mathrm{Tx},k}}$ hold
    so there are two rows of $\bm J_{{\bm \chi}}$ \chc{that} are linearly dependent on the others, which concludes the proof. 
\end{proof}

In contrast to \cite{li2023fpi}, \rev{herein,} when Eve knows {the structure of the designed precoder}, she still cannot distinguish the shifts introduced in the DAIS scheme from the true delay and angle information, as indicated by Proposition \ref{prop:strucleakge}, which suggests the robustness of our scheme.

\vspace{-5pt}
\section{Numerical Results}\label{sec:numericalresults}\vspace{-3pt}
In this section, to show that the proposed DAIS scheme effectively protects Alice's location from being accurately estimated by Eve, we numerically evaluate the lower bound of Eve's localization error derived in Section \ref{sec:eveerror}, which is the best performance that Eve can achieve with a misspecified-unbiased estimator. In addition, the derived CRB for Bob's localization error is also provided.

In all of the numerical results, the parameters $K$, \chc{$\varphi_c$}, $B$, $c$,  $N_t$, $N$, $G$ are set to $2$, $60$ GHz, $30$ MHz, $300$ m/us, $16$, $16$, and $16$, respectively. For the adopted channel model in \eqref{eq:channel_subcarrier}, channel coefficients are numerically generated according to the free-space path loss model \cite{Goldsmith} while the scatterers of the two NLOS paths are at $[8.87\text{ m}, -6.05 \text{ m}]^{\mathrm{T}}$ and $[7.44 \text{ m}, 8.53 \text{ m}]^{\mathrm{T}}$, respectively. Alice is located at $[3 \text{ m},0 \text{ m}]^{\mathrm{T}}$, transmitting certain random, complex values uniformly generated on the unit circle as the pilot signals.  For a fair comparison, we place Bob and Eve at the same location $[10 \text{ m},5 \text{ m}]^{\mathrm{T}}$ so that the same received signals are used for their individual localization. The evaluated root-mean-square errors (RMSE) for Bob's and Eve’s localization are defined as \vspace{-10pt}
\begin{equation}
    \begin{aligned}
        \operatorname{RMSE}_{\text{Bob}}&\triangleq\sqrt{\bm\Xi_{\bm \phi^{\star}}[1,1]+\bm\Xi_{\bm \phi^{\star}}[2,2]},\\
        \operatorname{RMSE}_{\text{Eve}}&\triangleq\sqrt{\bm\Psi_{\bar{\bm \phi}^\star}[1,1]+\bm\Psi_{\bar{\bm \phi}^\star}[2,2]}, 
    \end{aligned}\vspace{-5pt}
\end{equation}
respectively. Unless otherwise stated, the SNR is defined as $\operatorname{SNR}\triangleq10\log_{10}\frac{\sum^{G}_{g=1}\sum^{N-1}_{n=0}|{\bar{u}}^{(g,n)}|^{2}}{NG\sigma^2}$.

\begin{figure}[t]
\centering
\includegraphics[scale=0.48]{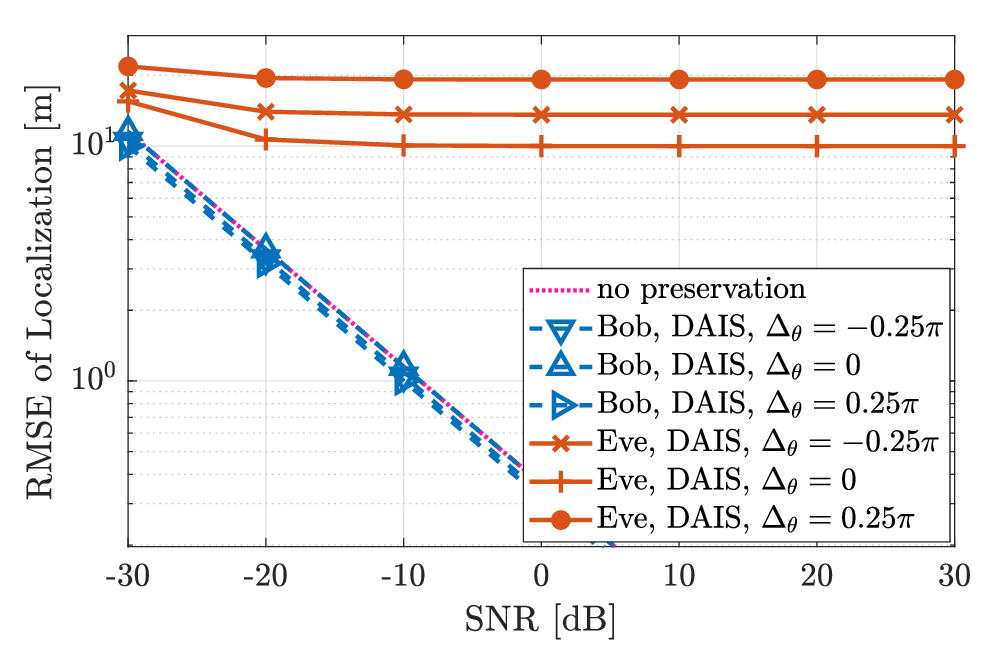}\vspace{-11pt}
\caption{{Lower bounds for the RMSE of localization with different choices of $\Delta_\theta$, where $\Delta_\tau=T_s$.}}
\label{fig:rmse_angle} \vspace{-15pt}
\end{figure}
The RMSEs for Bob's and Eve's localization accuracy are shown in Figure \ref{fig:rmse_angle}, where $\Delta_\tau$ is fixed at $T_s$ while $\Delta_\theta$ is set to $-0.25\pi$, $0$, and $0.25\pi$, respectively, for the proposed DAIS scheme. Since Bob can receive the shared information through a secure channel and construct the effective pilot signal $\bar{\bm s}^{(g,n)}$ for his localization, the obfuscation caused by the proposed DAIS scheme can be removed, leading to negligible loss according to Figure \ref{fig:rmse_angle}. In contrast, due to lack of the knowledge of $\Delta_\tau$ and $\Delta_\theta$, as shown in Figure \ref{fig:rmse_angle}, there is a strong degradation of Eve's localization accuracy.  To be more specific, when $\operatorname{SNR}$ is $0$ dB and $\Delta_\theta$ is set to $0.25\pi$, $\operatorname{RMSE}_{\text{Eve}}$ is up to around $19.22$ m because of the introduced geometric mismatch, while $\operatorname{RMSE}_{\text{Bob}}$ can be maintained at around $0.32$ m. Coinciding with the analysis provided in Section \ref{sec:degradationdais}, at high SNRs, Eve’s localization accuracy is mainly affected by the distance between Alice's true location and the corresponding pseudo-true location\footnote{In terms of the reduction of Eve's localization accuracy, the optimal choice of $\Delta_\theta$ and $\Delta_\tau$ is unknown without CSI but it can be proved that for a given $\Delta_\tau$, such a distance with $\sin(\Delta_\theta)\neq 0$ is greater than that with $\sin(\Delta_\theta)=0$.}; such distances are around $13.61$ m, $10.00$ m, and $19.22$ m for $\Delta_\theta=-0.25\pi$, $\Delta_\theta=0$, and $\Delta_\theta=0.25\pi$, respectively, which leads to distinct increases of localization error.

\begin{figure}[t]
\centering
\includegraphics[scale=0.48]{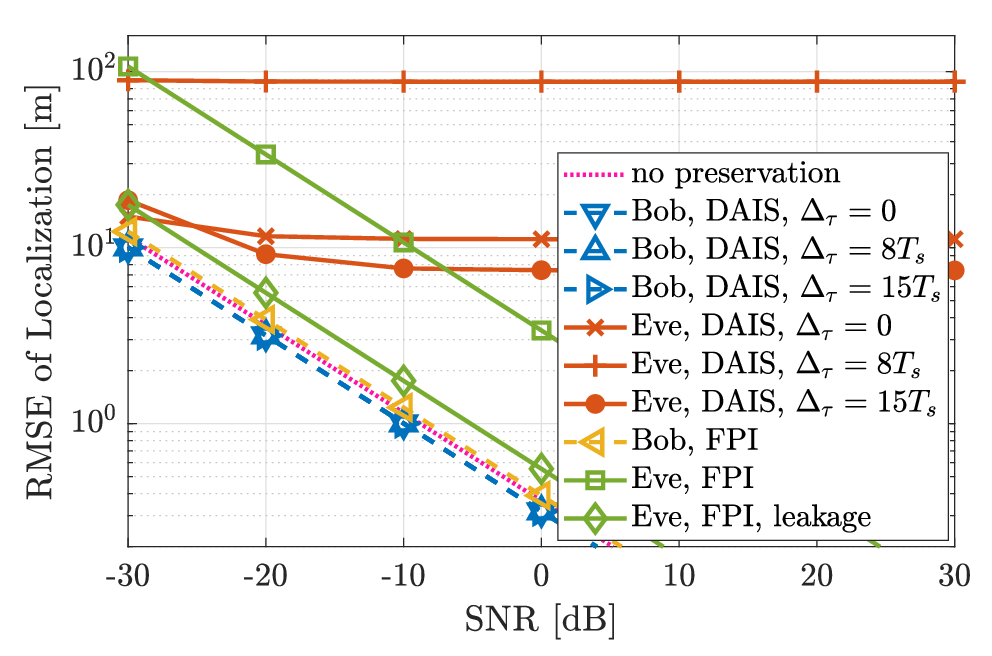}\vspace{-10pt}
\caption{\chc{Localization RMSE for different choices of $\Delta_\tau$ and $\Delta_\theta=0.25\pi$;  the design parameters $\bar\delta_\tau$ and $\bar\delta_{\theta_{\text{TX}}}$ used for the FPI scheme are set according to \cite{li2023fpi}.  {The DAIS and FPI precoders are distinctly different.}}}
\label{fig:rmse_delay}\vspace{-15pt}
\end{figure}

To strengthen the location-privacy enhancement, the shift for the TOAs can be also adjusted, moving pseudo-true location further away from the Alice's true location. As shown in Figure \ref{fig:rmse_delay}, for a given $\Delta_\theta=0.25\pi$, Eve’s localization error can be increased to $87.66$ m at $\operatorname{SNR}=0$ dB if $\Delta_\tau=8T_s$. We note that, similar to the choice of $\Delta_\theta$, the lower bound for Eve's localization error is not monotonically increasing with respect to $\Delta_\theta$ due to phase wrapping. From Figure \ref{fig:rmse_delay}, it can be observed that the largest value of $\Delta_\tau$ in the numerical results, \textit{i.e.,} $\Delta_\tau=15T_s$, does not result in the worst localization accuracy for Eve, where a NLOS path is incorrectly used as the LOS path for localization, yet there is a more than $15$ dB gap as compared with Bob who knows the shared information when $\operatorname{SNR}$ is higher than $-10$ dB\footnote{For the low SNRs, the effect of the geometric mismatch is relatively less significant due to the noise. The exact performance is also affected by other factors, \textit{e.g.,} AODs $\bar{\bm\theta}_{\mathrm{Tx}}$, according to the analysis in \cite{li2023fpi}.}. Furthermore, as compared with the FPI scheme \cite{li2023fpi}, the proposed DAIS design results in a comparable accuracy degradation for Eve\footnote{For a more comprehensive understanding of Eve's localization accuracy with the proposed DAIS in practice, we will evaluate the performance of the specific estimators in the future work.}, with the reduced sensitivity to the {leakage of the precoder structure}.

\vspace{-5pt}
\section{Conclusions}\vspace{-2.5pt}
Location-privacy was enhanced by obfuscating the eavesdropper with a \textit{delay-angle information spoofing} design. A CSI-free framework was proposed for {DAIS}, where the location-relevant delays and angles were shifted, misleading the eavesdropper into estimating \rev{an incorrect} physical position. To this end, a {precoder} was designed. By leveraging the securely shared information (only two parameters), the introduced obfuscation can be \rev{removed} for the legitimate localizer. Theoretical analysis of the localization error with DAIS was provided, validating the efficacy of the proposed scheme. Furthermore, the leakage of the {precoder structure} was analyzed, indicating the further robustness of DAIS relative to the {previously proposed} FPI scheme. With respect to localization accuracy, there was more than $15$ dB accuracy degradation for the eavesdropper at high SNRs.

\vspace{-5pt}

\renewcommand*{\bibfont}{\footnotesize}
\printbibliography
\end{document}